\documentclass[11pt]{article}

\usepackage[margin=1in]{geometry}
\usepackage{amsmath,amssymb,amsthm,amsfonts}
\usepackage{bm}
\usepackage{mathtools}
\usepackage{physics}
\usepackage{graphicx}
\usepackage{placeins}
\usepackage{microtype}
\usepackage{hyperref}
\usepackage[nameinlink,capitalise]{cleveref}

\allowdisplaybreaks[2]

\graphicspath{{figures/}}

\hypersetup{
  colorlinks=true,
  linkcolor=blue,
  citecolor=blue,
  urlcolor=blue,
  pdftitle={Hyperbolic Completion of Newton's Off-Center Orbit Problem: Constant-Curvature Rigidity, SO(2,1) Symmetry, and Magnetic Uniqueness},
  pdfauthor={Dipesh Bhandari}
}

\title{Hyperbolic Completion of Newton's Off-Center Orbit Problem:\\
Constant-Curvature Rigidity, \texorpdfstring{$SO(2,1)$}{SO(2,1)} Symmetry, and Magnetic Uniqueness}
\author{
Dipesh Bhandari\thanks{Corresponding author:
\href{mailto:dbhandari@smu.edu}{dbhandari@smu.edu}}\\
Department of Physics, Southern Methodist University\\
Dallas, Texas 75275, USA
}
\date{}

\newtheorem{theorem}{Theorem}[section]
\newtheorem{proposition}[theorem]{Proposition}
\newtheorem{corollary}[theorem]{Corollary}

\theoremstyle{remark}
\newtheorem{remark}[theorem]{Remark}

\newcommand{\HH}{\mathbb H}
\newcommand{\eez}{\bm e_z}
\newcommand{\rr}{\bm r}
\newcommand{\pp}{\bm p}
\newcommand{\bPi}{\bm\Pi}
\newcommand{\crossz}[1]{\eez\times #1}

\begin{document}

\maketitle

\begin{abstract}
The hyperbolic potential
\[
V(r)=-\frac{\alpha}{(R^2-r^2)^2}
\]
was identified by Olshanii as the negative-curvature counterpart of Newton's
off-center circular-orbit problem.  We first show that this form is not an
isolated ansatz.  For a planar radial natural Hamiltonian at fixed energy, if
the Jacobi metric is regular at the force center and has nonzero constant
Gaussian curvature, then, up to radial and overall scales,
\[
E-V(r)=\frac{\alpha}{(R^2\pm r^2)^2}.
\]
Thus the spherical and hyperbolic off-center models are the two nonflat
regular radial constant-curvature branches.

We then complete the dynamics of the negative-curvature branch.  At zero
energy its Jacobi metric is a constant multiple of the Poincare disk metric.
An explicit Runge--Lenz-type vector and the angular momentum form an on-shell
$\mathfrak{so}(2,1)$ moment map; their algebra gives the supporting-circle
equation without integrating the equations of motion.  The cotangent lift of
circular inversion preserves the moment map and intertwines the exterior and
punctured-interior flows after a positive time change.  The singular circle is
at infinite Jacobi distance but is reached in finite Newtonian time.

Finally, requiring the magnetic two-form to preserve the full
orientation-preserving hyperbolic isometry group uniquely forces it to be a
constant multiple of hyperbolic area, hence
$B(r)=-Q/(R^2-r^2)^2$ in disk coordinates.  The resulting shifted
$\mathfrak{so}(2,1)$ Casimir gives the standard
circle--horocycle--hypercycle transition at $Q^2=8m\alpha R^2$, now as an
explicit coupling statement in the off-center Newtonian system.
\end{abstract}

\medskip
\noindent\textbf{Keywords:}
off-center orbits; hyperbolic geometry; integrable Hamiltonian systems;
dynamical symmetry; circular inversion; magnetic flows; singular potentials.

\medskip
\noindent\textbf{2020 Mathematics Subject Classification:}
Primary 37J35, 70H06;
Secondary 37J37, 53C22, 81Q10.
\section{Introduction}
\label{sec:introduction}

A central force singles out the origin, whereas a circular trajectory has its
own Euclidean center.  Newton's off-center-orbit question asks when these two
centers can differ.  Olshanii showed that one zero-energy central-force model
has displaced circular orbits whose Jacobi metric is spherical, and he also
wrote down the complementary singular potential expected to realize the
hyperbolic branch \cite{Olshanii}.  The spherical magnetic continuation was
subsequently related to monopole dynamics on the sphere
\cite{BhandariCrescimanno}.

The present paper starts from the hyperbolic potential
\begin{equation}
H(\rr,\pp)=\frac{\pp^2}{2m}-\frac{\alpha}{(R^2-r^2)^2},
\qquad r=|\rr|,\qquad \alpha>0,
\label{eq:H}
\end{equation}
and asks which statements about \eqref{eq:H} are genuinely structural rather
than consequences of having recognized the Poincare disk.

That distinction matters.  The potential in \eqref{eq:H} and the fact that its
interior zero-energy trajectories are Poincare geodesics were already
anticipated in Ref.~\cite{Olshanii}.  Nor are constant-field magnetic circles,
horocycles, and hypercycles on $\mathbb H^2$ new
\cite{ComtetHouston,Comtet,BarrosLandau}.  Our contribution begins after those
observations.  We prove a rigidity statement showing why the spherical and
hyperbolic potentials are the two regular radial constant-curvature Jacobi
branches; construct the hyperbolic isometry moment map explicitly in the
original Newtonian phase-space variables; treat the exterior region and its
canonical inversion into the punctured disk; determine the Newtonian-time
behavior at the singular boundary; and prove that full $SO^+(2,1)$ invariance
removes the apparent freedom in the magnetic ansatz.

Conformal and inversion transformations in integrable mechanics have a much
broader history.  In particular, Albouy and Zhao give a modern account of
Darboux inversion for Kepler systems on constant-curvature surfaces
\cite{AlbouyZhao}.  The inversion result below is more specific: it is an
explicit cotangent-lift identity for the singular Hamiltonian
\eqref{eq:H}, including its conserved moment map and its zero-energy flow.

The first result explains why the potential itself is natural.

\begin{theorem}[Radial constant-curvature rigidity]
\label{thm:intro-rigidity}
Let
\[
H=\frac{\pp^2}{2m}+V(r)
\]
be a planar radial natural Hamiltonian, fix an energy $E$, and suppose
$E-V(r)>0$ near $r=0$.  Assume that the Jacobi metric
\[
g_J=2m(E-V(r))(dx^2+dy^2)
\]
extends smoothly and nondegenerately across $r=0$ and has constant Gaussian
curvature $\mathcal K_J\neq0$.  Then there are $R>0$ and $\alpha>0$ such that
\begin{equation}
E-V(r)=\frac{\alpha}{(R^2+\varepsilon r^2)^2},
\qquad
\varepsilon=\operatorname{sgn}(\mathcal K_J),
\label{eq:intro-rigidity}
\end{equation}
with
\begin{equation}
\mathcal K_J=\varepsilon\frac{2R^2}{m\alpha}.
\label{eq:intro-curvature}
\end{equation}
For negative curvature the maximal regular component containing the origin is
$r<R$; for positive curvature the formula is regular on the whole plane.
\end{theorem}

The ordinary differential equation behind
\cref{thm:intro-rigidity} is the radial Liouville equation, so we do not present
the classification of constant-curvature conformal metrics itself as new.
Its role here is conceptual: within regular radial Jacobi mechanics, the
spherical denominator $R^2+r^2$ and the hyperbolic denominator $R^2-r^2$ are
forced rather than guessed.

For the negative-curvature branch define
\[
D_R=\{\rr:r<R\},\qquad E_R=\{\rr:r>R\}.
\]
The main dynamical statement is summarized as follows.

\begin{theorem}[Hyperbolic completion]
\label{thm:intro-hyperbolic}
On the zero-energy hypersurface of \eqref{eq:H}:
\begin{enumerate}
\item every nonradial trajectory in either $D_R$ or $E_R$ lies on a Euclidean
circle orthogonal to $r=R$, while radial trajectories lie on lines through the
origin;
\item the angular momentum $L_z$ and the vector $\bm K$ defined in
\cref{sec:symmetry} form an on-shell $\mathfrak{so}(2,1)$ moment map whose
quadratic Casimir is the hyperbolic geodesic Hamiltonian;
\item the cotangent lift of circular inversion $\rr\mapsto R^2\rr/r^2$
preserves $L_z$ and $\bm K$ and intertwines the exterior flow with the
punctured-interior flow up to a positive reparametrization of time;
\item the singular circle $r=R$ lies at infinite Jacobi distance but is reached
in finite Newtonian time.
\end{enumerate}
\end{theorem}

The magnetic field is also fixed once the symmetry requirement is stated
precisely.

\begin{theorem}[Uniqueness of the fully invariant magnetic deformation]
\label{thm:intro-magnetic-unique}
Let $g_{\mathbb H}$ be the Poincare metric on $D_R$.  If a smooth magnetic
two-form $\mathcal B$ is invariant under the full orientation-preserving
isometry group $\mathrm{Isom}^+(D_R,g_{\mathbb H})\simeq SO^+(2,1)$, then
\begin{equation}
\mathcal B=b\,d\mu_{\mathbb H}
\label{eq:intro-magnetic-area}
\end{equation}
for a constant $b$.  In Euclidean disk coordinates this is exactly
\begin{equation}
B(r)=\frac{4bR^4}{(R^2-r^2)^2}
      =-\frac{Q}{(R^2-r^2)^2}
\label{eq:intro-magnetic-field}
\end{equation}
for one constant coupling $Q$.
\end{theorem}

Thus the magnetic field used below is not selected by trial.  It is the unique
smooth two-form, up to overall strength, compatible with the full hyperbolic
isometry symmetry.  Once this field is inserted, the standard hyperbolic
Landau trichotomy appears in the original Newtonian variables and the critical
coupling becomes $Q^2=8m\alpha R^2$.

The paper is organized around these three statements.  \Cref{sec:rigidity}
proves the radial constant-curvature classification.  \Cref{sec:jacobi}
specializes it to the hyperbolic branch.  \Cref{sec:symmetry,sec:orbits}
derive the phase-space moment map and the orbit equation.
\Cref{sec:inversion,sec:collision} treat the exterior component, inversion, and
the singular boundary.  \Cref{sec:magnetic} proves magnetic uniqueness and
develops the shifted moment map and orbit classification.  A short quantum
side remark is kept in \cref{app:quantum}; it is not part of the main claim.

\section{Why the potential is forced: radial constant-curvature Jacobi metrics}
\label{sec:rigidity}

At fixed energy, a planar natural Hamiltonian
\begin{equation}
H=\frac{\pp^2}{2m}+V(r)
\end{equation}
has Jacobi metric
\begin{equation}
g_J=2m(E-V(r))(dx^2+dy^2)=e^{2\sigma(r)}(dx^2+dy^2).
\label{eq:general-jacobi}
\end{equation}
For a radial conformal factor, constant Gaussian curvature $\mathcal K$ is the
Liouville equation
\begin{equation}
\sigma''(r)+\frac1r\sigma'(r)+\mathcal K e^{2\sigma(r)}=0.
\label{eq:radial-liouville}
\end{equation}
The regularity condition at the force center removes the singular radial
solutions and leaves only the stereographic branches.

We now prove \cref{thm:intro-rigidity}.  Under its hypotheses the conformal
factor must be
\begin{equation}
e^{2\sigma(r)}=
\begin{cases}
\displaystyle \frac{4R^2}{\mathcal K_J(R^2+r^2)^2},
&\mathcal K_J>0,\\[2.0ex]
\displaystyle \frac{4R^2}{|\mathcal K_J|(R^2-r^2)^2},
&\mathcal K_J<0,
\end{cases}
\label{eq:classified-conformal-factor}
\end{equation}
for some $R>0$.

\begin{proof}[Proof of \cref{thm:intro-rigidity}]
Set $t=\log r$ and
\begin{equation}
u(t)=\sigma(e^t)+t.
\end{equation}
Then \eqref{eq:radial-liouville} becomes
\begin{equation}
u''+\mathcal K_J e^{2u}=0.
\label{eq:u-liouville}
\end{equation}
Smooth nondegeneracy at the origin means
$e^{\sigma(r)}\to\lambda_0>0$ and $r\sigma'(r)\to0$ as $r\downarrow0$.
Hence, as $t\to-\infty$,
\begin{equation}
u(t)=t+\log\lambda_0+o(1),
\qquad u'(t)\to1.
\end{equation}
Multiplying \eqref{eq:u-liouville} by $2u'$ gives
\begin{equation}
(u')^2+\mathcal K_J e^{2u}=1.
\label{eq:u-first-integral}
\end{equation}
If $\mathcal K_J>0$, the solution with these asymptotics is
\begin{equation}
e^u=\frac{1}{\sqrt{\mathcal K_J}\cosh(t-t_0)}.
\end{equation}
Writing $R=e^{t_0}$ and dividing by $r=e^t$ gives
\begin{equation}
e^{\sigma(r)}=\frac{2R}{\sqrt{\mathcal K_J}(R^2+r^2)}.
\end{equation}
If $\mathcal K_J=-\kappa<0$, the regular solution on the component containing
the origin is
\begin{equation}
e^u=\frac{1}{\sqrt\kappa\sinh(t_0-t)},\qquad t<t_0,
\end{equation}
which yields
\begin{equation}
e^{\sigma(r)}=\frac{2R}{\sqrt\kappa(R^2-r^2)},\qquad 0\le r<R.
\end{equation}
Squaring and using $e^{2\sigma}=2m(E-V)$ proves the result.
\end{proof}

\begin{corollary}[The two nonflat regular radial branches]
\label{cor:two-branches}
At $E=0$, every regular radial Jacobi metric of nonzero constant curvature is,
up to the two positive scales $R$ and $\alpha$,
\begin{equation}
V_{+}(r)=-\frac{\alpha}{(R^2+r^2)^2},
\qquad
V_{-}(r)=-\frac{\alpha}{(R^2-r^2)^2}.
\label{eq:two-potentials}
\end{equation}
The first has positive Jacobi curvature and the second negative Jacobi
curvature.  For $\mathcal K_J=0$, regularity at the origin instead forces the
Jacobi conformal factor to be constant, hence $V$ is constant at the chosen
energy scale.
\end{corollary}

\begin{remark}[Scope]
The theorem classifies potentials only under the stated fixed-energy,
radial, regular-center, constant-curvature Jacobi assumptions.  It does not
classify all central forces with special closed orbits, nor singular Liouville
metrics with conical behavior at the origin.  Within this restricted problem,
however, the two denominators in \eqref{eq:two-potentials} are exhaustive.
\end{remark}

\section{The hyperbolic branch and the Poincare disk}\label{sec:jacobi}

We now specialize \cref{cor:two-branches} to the negative-curvature potential
\eqref{eq:H}.  At zero energy,
\begin{equation}
 ds_J^2=\frac{2m\alpha}{(R^2-r^2)^2}(dx^2+dy^2).
\label{eq:jacobi}
\end{equation}
On $D_R$ compare this with the Poincare disk metric
\begin{equation}
 ds_{\HH^2_R}^2=\frac{4R^4}{(R^2-r^2)^2}(dx^2+dy^2),
\label{eq:poincare}
\end{equation}
whose Gaussian curvature is $-1/R^2$.  The two metrics differ only by the
constant factor
\begin{equation}
 ds_J^2=\frac{m\alpha}{2R^4}\,ds_{\HH^2_R}^2,
\label{eq:constant-rescale}
\end{equation}
so
\begin{equation}
\boxed{\mathcal K_J=-\frac{2R^2}{m\alpha}}.
\label{eq:jacobi-curvature}
\end{equation}
They therefore have the same Levi--Civita connection and the same
unparametrized geodesics.  Every interior zero-energy trajectory is either a
Euclidean circle orthogonal to $r=R$ or a diameter of the disk
\cite{Helgason}.

The same conformal expression is defined for $r>R$, but that component is
separated from the disk by the singular circle.  Its precise relation to the
punctured disk is established by the canonical inversion in
\cref{sec:inversion}.

\section{Hidden symmetry: the \texorpdfstring{$\mathfrak{so}(2,1)$}{so(2,1)} moment map}\label{sec:symmetry}

The Poincar\'e disk description predicts three infinitesimal isometries:
one rotation and two noncompact translations.  We now recover their
phase-space moment map directly in the Newtonian variables.  Let
\begin{equation}
 L_z=xp_y-yp_x
\label{eq:Lz}
\end{equation}
be the angular momentum.  Rotational invariance gives
$\{L_z,H\}=0$.  Define the vector
\begin{equation}
 \boxed{
 \bm K
 =L_z\rr+(\rr\cdot\pp)\,\crossz{\rr}
 -R^2\crossz{\pp}.
 }
\label{eq:K}
\end{equation}
In components,
\begin{align}
 K_x&=(R^2+x^2-y^2)p_y-2xy\,p_x,
\label{eq:Kx}\\
 K_y&=(x^2-y^2-R^2)p_x+2xy\,p_y.
\label{eq:Ky}
\end{align}

\begin{proposition}[Conformal integral]
\label{prop:K-conformal}
The vector \eqref{eq:K} satisfies
\begin{equation}
 \boxed{
 \{\bm K,H\}=4H\,\crossz{\rr}.
 }
\label{eq:K-H}
\end{equation}
It is therefore conserved on the zero-energy hypersurface.
\end{proposition}

\begin{proof}
Using $\{x_i,p_j\}=\delta_{ij}$ and differentiating
\eqref{eq:Kx}--\eqref{eq:Ky}, one obtains
\begin{equation}
 \{K_x,H\}=-4yH,
 \qquad
 \{K_y,H\}=4xH,
\end{equation}
which is equivalent to \eqref{eq:K-H}.
\end{proof}

\begin{proposition}[$\mathfrak{so}(2,1)$ algebra]
\label{prop:so21}
The generators satisfy
\begin{align}
 \{K_x,K_y\}&=-4R^2L_z,
\label{eq:KK}\\
 \{K_x,L_z\}&=-K_y,
\label{eq:KxL}\\
 \{K_y,L_z\}&=K_x.
\label{eq:KyL}
\end{align}
Equivalently, with
\begin{equation}
 J_0=L_z,
 \qquad
 J_1=\frac{K_x}{2R},
 \qquad
 J_2=\frac{K_y}{2R},
\end{equation}
we have
\begin{equation}
 \{J_0,J_1\}=J_2,
 \qquad
 \{J_0,J_2\}=-J_1,
 \qquad
 \{J_1,J_2\}=-J_0.
\label{eq:so21}
\end{equation}
\end{proposition}

\begin{proof}
The relations follow by direct evaluation of the canonical Poisson brackets.
The sign in \eqref{eq:KxL} is fixed by the convention
$\epsilon_{xy}=+1$ and by the fact that $\{L_z,K_x\}=K_y$.
\end{proof}

\begin{proposition}[Casimir and hyperbolic kinetic energy]
\label{prop:casimir}
The quadratic Casimir is
\begin{equation}
 \mathcal C
 =\frac{\bm K^2}{4R^2}-L_z^2.
\label{eq:casimir}
\end{equation}
It admits the identity
\begin{equation}
 \boxed{
 \mathcal C
 =\frac{(R^2-r^2)^2\pp^2}{4R^2}
 =2R^2H_{\mathrm{geo}},
 }
\label{eq:casimir-geodesic}
\end{equation}
where
\begin{equation}
 H_{\mathrm{geo}}
 =\frac12g^{ij}_{\HH^2_R}p_ip_j
 =\frac{(R^2-r^2)^2\pp^2}{8R^4}
\label{eq:Hgeo}
\end{equation}
 is the geodesic Hamiltonian for \eqref{eq:poincare}.  On the zero-energy
hypersurface of \eqref{eq:H},
\begin{equation}
 \boxed{
 \mathcal C=\frac{m\alpha}{2R^2}.
 }
\label{eq:casimir-shell}
\end{equation}
\end{proposition}

\begin{proof}
Substitution of \eqref{eq:K} into \eqref{eq:casimir} gives the first identity
in \eqref{eq:casimir-geodesic}.  The second follows from the inverse metric of
\eqref{eq:poincare}.  Finally, $H=0$ implies
\begin{equation}
 (R^2-r^2)^2\pp^2=2m\alpha,
\end{equation}
which yields \eqref{eq:casimir-shell}.
\end{proof}

Thus the Runge--Lenz-type quantities are not accidental constants.  They
are the moment-map components of the noncompact isometry group of the
hyperbolic Jacobi geometry.  The orbit shape can therefore be read from a
coadjoint invariant rather than obtained by solving a radial equation.

\section{From conserved quantities to the orbit circles}\label{sec:orbits}

The conserved moment map gives the Euclidean orbit equation directly, without
integrating the equations of motion.

\begin{proposition}[Algebraic orbit equation]
\label{prop:orbit-equation}
The generators satisfy the identity
\begin{equation}
 \boxed{
 \bm K\cdot\rr=L_z(r^2+R^2).
 }
\label{eq:Kdotr}
\end{equation}
For $L_z\neq0$, every zero-energy trajectory therefore lies on the circle
\begin{equation}
 r^2-2\bm a\cdot\rr+R^2=0,
 \qquad
 \bm a=\frac{\bm K}{2L_z}.
\label{eq:orbit-circle}
\end{equation}
\end{proposition}

\begin{proof}
Taking the scalar product of \eqref{eq:K} with $\rr$ gives
\begin{align}
 \bm K\cdot\rr
 &=L_zr^2
 +(\rr\cdot\pp)(\crossz{\rr})\cdot\rr
 -R^2(\crossz{\pp})\cdot\rr \\
 &=L_zr^2+R^2L_z,
\end{align}
where $(\crossz{\rr})\cdot\rr=0$ and
$(\crossz{\pp})\cdot\rr=-L_z$.  Dividing by $L_z$ yields
\eqref{eq:orbit-circle}.
\end{proof}

Completing the square gives
\begin{equation}
 |\rr-\bm a|^2=|\bm a|^2-R^2.
\label{eq:completed-circle}
\end{equation}
Hence
\begin{equation}
 \rr_c=\bm a,
 \qquad
 \rho^2=|\bm a|^2-R^2.
\label{eq:center-radius}
\end{equation}
The Casimir gives the more explicit zero-energy relations
\begin{equation}
 |\bm a|^2
 =R^2+\frac{m\alpha}{2L_z^2},
 \qquad
 \boxed{
 \rho^2=\frac{m\alpha}{2L_z^2}.
 }
\label{eq:radius-L}
\end{equation}

\begin{corollary}[Orthogonality and exterior force center]
\label{cor:orthogonal}
The supporting circle \eqref{eq:orbit-circle} intersects $r=R$
orthogonally, and the origin lies outside the supporting circle.
\end{corollary}

\begin{proof}
Two Euclidean circles with center separation $d$, radii $R$ and $\rho$, meet
orthogonally precisely when $d^2=R^2+\rho^2$.  Here
$d=|\bm a|$, so the condition is exactly \eqref{eq:center-radius}.  Moreover,
$|\bm a|^2-\rho^2=R^2>0$, hence $|\bm a|>\rho$ and the origin lies outside the
orbit circle.
\end{proof}

This is the direct geometric realization of the hyperbolic off-center
case.  In the spherical predecessor the force center is enclosed by each
supporting circle; here it is excluded.  The relation
$|\rr_c|^2=R^2+\rho^2$ is therefore the geometric signature that distinguishes
the hyperbolic branch.

The complete geometry is shown in \cref{fig:exact-geometry}.  The plotted
orbit is constructed from \eqref{eq:orbit-circle}; the right-angle markers
make the orthogonality relation explicit, while the relative positions of the
force center and the supporting circle display the physical ``off-center''
content of \cref{thm:intro-hyperbolic}.

\begin{figure}[!tbp]
\centering
\includegraphics[width=0.94\linewidth]{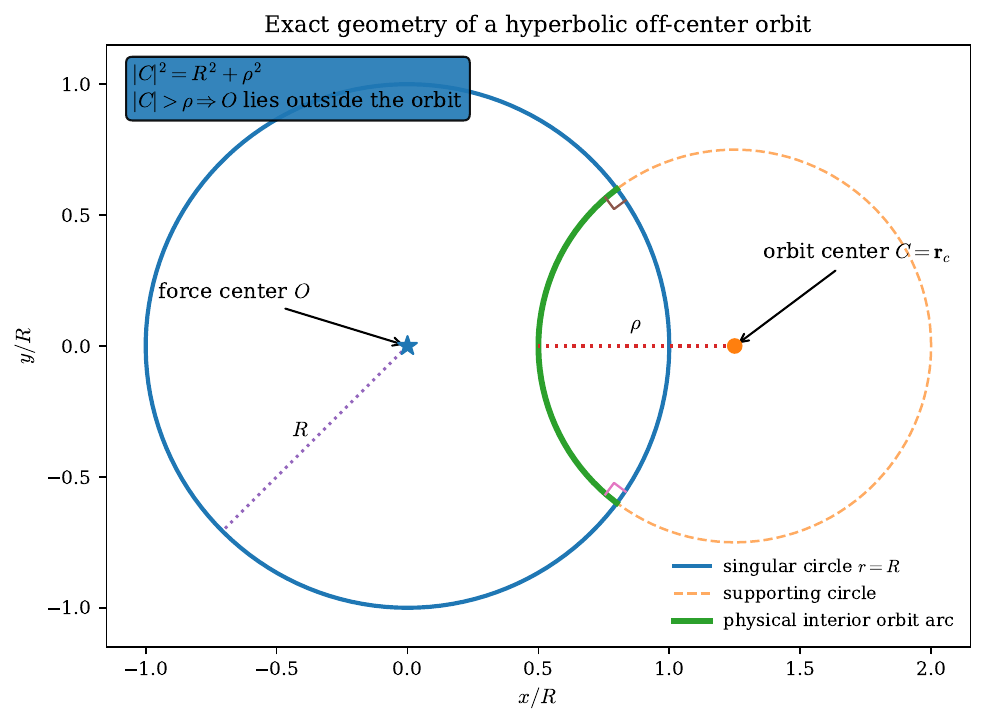}
\caption{Exact geometry of a nonradial zero-energy orbit.  The physical
interior trajectory is the part of a Euclidean supporting circle lying in
$r<R$.  Its center $C=\rr_c$ and radius $\rho$ obey
$|C|^2=R^2+\rho^2$, proving both orthogonality to the singular circle and
that the force center $O$ lies outside the supporting circle.}
\label{fig:exact-geometry}
\end{figure}

\begin{proposition}[Radial trajectories]
\label{prop:radial}
If $L_z=0$, then the trajectory lies on a fixed line through the origin.
Its interior branch is a diameter geodesic of the Poincar\'e disk; an exterior
branch is a radial ray on the same supporting line.
\end{proposition}

\begin{proof}
Conservation of $L_z$ and the identity $L_z=mr^2\dot\phi$ imply
$\dot\phi=0$ away from the origin.  Thus the polar angle is constant.  Diameter
lines are the $L_z=0$ geodesics of the disk model.
\end{proof}

\section{Inversion as an interior--exterior dynamical duality}
\label{sec:inversion}

The interior and exterior regions are disconnected in Newtonian
configuration space, but their orbit families are not independent.  Circular
inversion is the natural candidate for relating them.  What matters dynamically is its cotangent lift and its action on the Hamiltonian
flow, not only its action on the drawn orbit circles.

Consider circular inversion in the singular circle,
\begin{equation}
 \rr' =\mathcal I(\rr)
 =\frac{R^2}{r^2}\rr.
\label{eq:inversion-r}
\end{equation}
Its canonical cotangent lift is
\begin{equation}
 \boxed{
 \pp'
 =\frac{r^2}{R^2}
 \left(
 \pp-2\frac{\rr\cdot\pp}{r^2}\rr
 \right).
 }
\label{eq:inversion-p}
\end{equation}
Indeed, \eqref{eq:inversion-p} is
$\pp'=(D\mathcal I)^{-T}\pp$, so the map
$(\rr,\pp)\mapsto(\rr',\pp')$ is symplectic.

\begin{theorem}[Canonical inversion symmetry]
\label{thm:inversion}
Under \eqref{eq:inversion-r}--\eqref{eq:inversion-p},
\begin{equation}
 \boxed{
 H(\rr',\pp')=\frac{r^4}{R^4}H(\rr,\pp).
 }
\label{eq:H-inversion}
\end{equation}
Moreover,
\begin{equation}
 \boxed{
 L_z'=L_z,
 \qquad
 \bm K'=\bm K.
 }
\label{eq:LK-inversion}
\end{equation}
Consequently inversion maps the zero-energy hypersurface to itself and
exchanges exterior branches with their images in the punctured interior
$D_R\setminus\{0\}$, carrying the same values of the symmetry generators.
\end{theorem}

\begin{proof}
The transformation \eqref{eq:inversion-p} is a radial reflection followed by
a scale, and therefore
\begin{equation}
 \pp'^2=\frac{r^4}{R^4}\pp^2.
\end{equation}
Also
\begin{equation}
 R^2-r'^2
 =R^2-\frac{R^4}{r^2}
 =-\frac{R^2}{r^2}(R^2-r^2),
\end{equation}
so
\begin{equation}
 \frac{1}{(R^2-r'^2)^2}
 =\frac{r^4}{R^4}\frac{1}{(R^2-r^2)^2}.
\end{equation}
This proves \eqref{eq:H-inversion}.  Direct substitution into the definitions
\eqref{eq:Lz} and \eqref{eq:K} gives \eqref{eq:LK-inversion}.
\end{proof}

\begin{corollary}[Orbit equivalence of the zero-energy flows]
\label{cor:inversion-flow}
Let $\Phi$ denote the canonical inversion
\eqref{eq:inversion-r}--\eqref{eq:inversion-p} and let $X_H$ be the
Hamiltonian vector field.  Then
\begin{equation}
 \boxed{
 \Phi_*X_H
 =\frac{r^4}{R^4}X_H
 \qquad\text{on }H=0.
 }
\label{eq:inversion-vector-field}
\end{equation}
Consequently $\Phi$ maps every zero-energy integral curve to a zero-energy
integral curve after a positive reparametrization of time.
\end{corollary}

\begin{proof}
The map $\Phi$ is symplectic and involutive, so
\begin{equation}
 \Phi_*X_H=X_{H\circ\Phi}.
\end{equation}
By \eqref{eq:H-inversion}, $H\circ\Phi=fH$ with
$f=r^4/R^4$.  The Leibniz rule for Hamiltonian vector fields gives
\begin{equation}
 X_{fH}=fX_H+H X_f.
\end{equation}
Restriction to $H=0$ proves \eqref{eq:inversion-vector-field}.  Since $f>0$
on both components, if $z'(t)=\Phi(z(t))$ then the change of parameter
$d\tau/dt=f(z'(t))$ makes $z'(\tau)$ an ordinary integral curve of $X_H$.
\end{proof}

\begin{corollary}[Exterior geometry and the missing point]
\label{cor:exterior-punctured}
The inversion \eqref{eq:inversion-r} is an isometry from the exterior Jacobi
metric on $E_R$ onto the restriction of the interior Jacobi metric to the
punctured disk $D_R\setminus\{0\}$.  Under this isometry, the exterior end
$r=\infty$ corresponds to the omitted origin.
\end{corollary}

\begin{proof}
Euclidean inversion obeys
\begin{equation}
 d\rr'^2=\frac{R^4}{r^4}d\rr^2.
\end{equation}
Combining this with the transformation of $(R^2-r^2)^2$ gives
\begin{equation}
 \frac{d\rr'^2}{(R^2-r'^2)^2}
 =\frac{d\rr^2}{(R^2-r^2)^2}.
\end{equation}
\end{proof}

\begin{remark}
The theorem and \cref{cor:inversion-flow,cor:exterior-punctured} supply an
exact duality between the exterior system and the punctured interior system,
including their zero-energy trajectories and reparametrized flows, but they do
not define a physical passage through $r=R$.
The Newtonian vector field is singular there.  The two branches are related by
inversion, rather than joined by an automatically determined collision law.
\end{remark}

The pointwise action of the inversion map is displayed in
\cref{fig:inversion}.  Interior points and their images lie on the same
supporting circle, in agreement with the exact invariance of $L_z$ and $\bm K$
in \eqref{eq:LK-inversion}.

\begin{figure}[!tbp]
\centering
\includegraphics[width=0.94\linewidth]{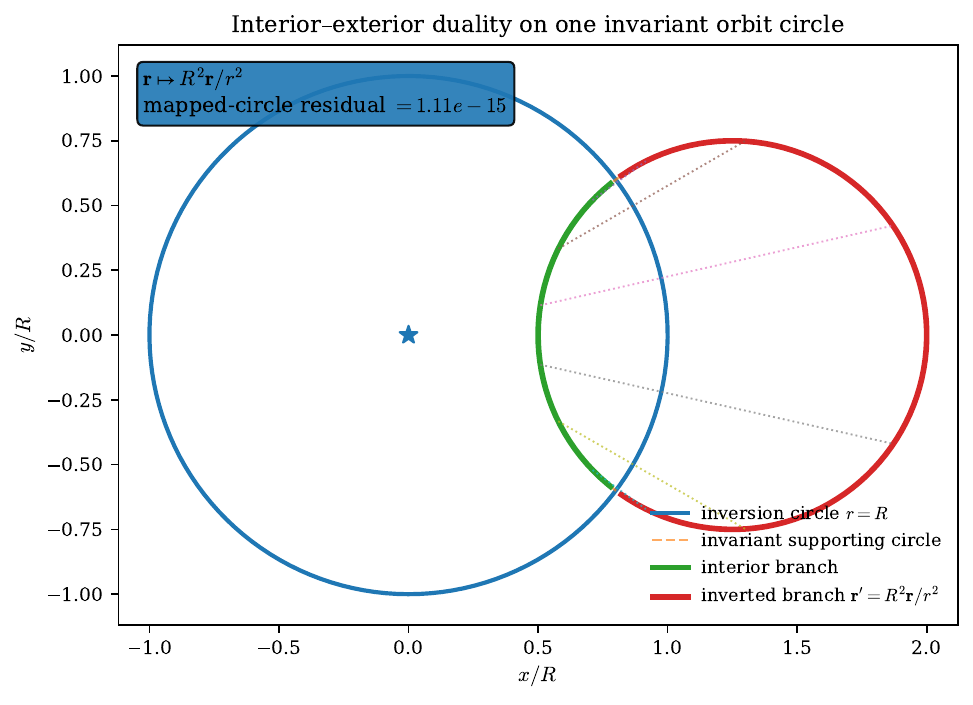}
\caption{Circular inversion maps the interior branch of an orthogonal
supporting circle to its exterior branch while leaving the same supporting
circle invariant.  Dotted radial segments join representative inversion
pairs.  The singular circle itself is not part of the Newtonian configuration
space.}
\label{fig:inversion}
\end{figure}

\section{What the singular boundary means physically}
\label{sec:collision}

The hyperbolic picture creates an apparent paradox.  The circle $r=R$ is
the ideal boundary of the Poincar\'e disk, yet the Newtonian force diverges
there.  The following results resolve the paradox by comparing the two time
and distance notions directly.

\begin{proposition}[Finite Newtonian collision time]
\label{prop:finite-time}
A zero-energy orbit that approaches the singular circle reaches it in finite
Newtonian time.  If
\begin{equation}
 \delta(t)=|R-r(t)|,
\end{equation}
then, as $t\uparrow t_*$,
\begin{equation}
 \boxed{
 t_*-t
 \sim
 R\sqrt{\frac{m}{2\alpha}}\,\delta(t)^2.
 }
\label{eq:collision-asymptotic}
\end{equation}
In particular,
\begin{equation}
 \delta(t)
 \sim
 \left(
 \sqrt{\frac{2\alpha}{m}}\frac{t_*-t}{R}
 \right)^{1/2}.
\end{equation}
\end{proposition}

\begin{proof}
The zero-energy condition gives
\begin{equation}
 \frac{m}{2}|\dot\rr|^2
 =\frac{\alpha}{(R^2-r^2)^2},
\end{equation}
so
\begin{equation}
 |\dot\rr|
 =\sqrt{\frac{2\alpha}{m}}
  \frac{1}{|R^2-r^2|}.
\label{eq:speed}
\end{equation}
Every nonradial disk geodesic meets $r=R$ orthogonally, and the same is true
for the radial geodesics.  Thus the leading normal velocity is the full speed.
Since
\begin{equation}
 |R^2-r^2|\sim2R\delta,
\end{equation}
we obtain
\begin{equation}
 \left|\frac{d\delta}{dt}\right|
 \sim
 \sqrt{\frac{2\alpha}{m}}\frac{1}{2R\delta}.
\end{equation}
Integration gives \eqref{eq:collision-asymptotic}.
\end{proof}

\begin{proposition}[Infinite hyperbolic distance]
\label{prop:infinite-distance}
The same singular circle lies at infinite distance in both the Poincar\'e and
Jacobi metrics.
\end{proposition}

\begin{proof}
Along a radial curve in the Poincar\'e metric,
\begin{equation}
 d_{\HH}(0,r)
 =\int_0^r\frac{2R^2}{R^2-s^2}\,ds
 =R\log\left(\frac{R+r}{R-r}\right),
\end{equation}
which diverges as $r\uparrow R$.  The Jacobi metric differs by the constant
factor in \eqref{eq:constant-rescale}, so its distance diverges as well.
\end{proof}

For interior trajectories, and for nonradial exterior trajectories whose two
ends approach $r=R$, the reparametrized hyperbolic geodesic is complete even
though the Newtonian flow terminates at the singular circle in finite time.
The exterior radial sector has a second and qualitatively different end at
$r=\infty$, which is the puncture identified in
\cref{cor:exterior-punctured}.

\begin{proposition}[Exterior radial escape]
\label{prop:exterior-infinity}
Along a zero-energy radial trajectory in $E_R$, the end $r=\infty$ is at
finite Jacobi distance but is reached only after infinite Newtonian time.  More
precisely, from any $r_0>R$,
\begin{equation}
 d_J(r_0,\infty)
 =\sqrt{2m\alpha}
  \int_{r_0}^{\infty}\frac{dr}{r^2-R^2}
 =\frac{\sqrt{2m\alpha}}{2R}
  \log\left(\frac{r_0+R}{r_0-R}\right)
 <\infty,
\label{eq:exterior-finite-jacobi-distance}
\end{equation}
whereas
\begin{equation}
 t(r)-t(r_0)
 =\sqrt{\frac{m}{2\alpha}}
  \int_{r_0}^{r}(s^2-R^2)\,ds
 \longrightarrow\infty
 \qquad (r\to\infty).
\label{eq:exterior-infinite-newton-time}
\end{equation}
In particular,
\begin{equation}
 r(t)\sim
 \left(3\sqrt{\frac{2\alpha}{m}}\,t\right)^{1/3}
 \qquad (t\to\infty),
\label{eq:exterior-radial-asymptotic}
\end{equation}
up to a translation of the Newtonian time origin.
\end{proposition}

\begin{proof}
For radial motion, the Jacobi line element on $E_R$ is
\begin{equation}
 ds_J=\frac{\sqrt{2m\alpha}}{r^2-R^2}\,dr,
\end{equation}
which gives \eqref{eq:exterior-finite-jacobi-distance}.  The zero-energy
condition gives
\begin{equation}
 \frac{dr}{dt}
 =\sqrt{\frac{2\alpha}{m}}\frac{1}{r^2-R^2}
\end{equation}
for the outgoing branch.  Integration yields
\eqref{eq:exterior-infinite-newton-time}; retaining the leading $r^3/3$ term
gives \eqref{eq:exterior-radial-asymptotic}.
\end{proof}

Thus the punctured exterior Jacobi geometry is geodesically incomplete only in
the radial direction toward $r=\infty$; adding the missing point corresponds,
under inversion, to restoring the origin of the full Poincar\'e disk.  By
contrast, the singular circle remains an infinite-distance ideal boundary.
The full Euclidean supporting circle should therefore be viewed as a geometric
completion of singular Newtonian branches, not as an automatically selected
physical continuation through $r=R$.


\section{The unique invariant magnetic deformation and its dynamics}
\label{sec:magnetic}

The nonmagnetic moment map comes from the full orientation-preserving isometry
group of the hyperbolic Jacobi metric.  Before writing magnetic brackets, we
therefore impose the corresponding symmetry requirement on the magnetic
two-form itself.  This turns the radial field from an ansatz into a
one-parameter classification.

\subsection{Uniqueness of the invariant field}

Let $G=\mathrm{Isom}^+(D_R,g_{\mathbb H})\simeq SO^+(2,1)$.  A magnetic system
on $T^*D_R$ uses the twisted symplectic form
\begin{equation}
\omega_{\mathcal B}=\omega_{\mathrm{can}}+\pi^*\mathcal B,
\end{equation}
where $\mathcal B$ is a two-form on the disk.

\begin{theorem}[Invariant magnetic-field rigidity]
\label{thm:magnetic-uniqueness}
A smooth two-form $\mathcal B$ on $D_R$ is invariant under $G$ if and only if
\begin{equation}
\mathcal B=b\,d\mu_{\mathbb H}
\label{eq:magnetic-area-unique}
\end{equation}
for a constant $b$.  In the Euclidean disk coordinates of
\eqref{eq:poincare}, this is
\begin{equation}
\boxed{B(r)=\frac{4bR^4}{(R^2-r^2)^2}=-\frac{Q}{(R^2-r^2)^2}}
\label{eq:magnetic-unique-B}
\end{equation}
with $Q=-4bR^4$.
\end{theorem}

\begin{proof}
On an oriented surface every two-form has the form
$\mathcal B=f\,d\mu_{\mathbb H}$ for a unique scalar function $f$.  The area
form is invariant under $G$.  If $\mathcal B$ is invariant, then
$f\circ g=f$ for every $g\in G$.  Because $G$ acts transitively on the
hyperbolic disk, $f$ is constant.  The converse is immediate.  Finally,
\begin{equation}
d\mu_{\mathbb H}=\frac{4R^4}{(R^2-r^2)^2}\,dx\wedge dy,
\end{equation}
which gives \eqref{eq:magnetic-unique-B}.
\end{proof}

Uniform magnetic fields on Riemannian surfaces are constant multiples of the
area form, and on hyperbolic surfaces they underlie the standard Landau and
horocycle dynamics \cite{ComtetHouston,Comtet,BarrosLandau,Paternain}.  The
point of \cref{thm:magnetic-uniqueness} is narrower: once preservation of the
full hyperbolic isometry group is required, no other position-dependent
magnetic two-form is available.

\subsection{Magnetic moment map}

We formulate the magnetic extension using gauge-covariant momenta
$\bPi=(\Pi_x,\Pi_y)$ and magnetic Poisson brackets
\begin{equation}
 \{x_i,x_j\}=0,
 \qquad
 \{x_i,\Pi_j\}=\delta_{ij},
 \qquad
 \{\Pi_x,\Pi_y\}=B(r).
\label{eq:magnetic-PB}
\end{equation}
The Hamiltonian retains the form
\begin{equation}
 H_B=\frac{\bPi^2}{2m}
 -\frac{\alpha}{(R^2-r^2)^2}.
\label{eq:HB}
\end{equation}

By \cref{thm:magnetic-uniqueness}, write
\begin{equation}
 \boxed{
 B(r)=-\frac{Q}{(r^2-R^2)^2},
 \qquad
 G(r)=\frac{Q}{r^2-R^2}.
 }
\label{eq:BG}
\end{equation}
These obey
\begin{equation}
 G'(r)=2rB(r).
\label{eq:Gprime}
\end{equation}
Define the conserved angular momentum candidate
\begin{equation}
 \boxed{
 L_z^{(B)}
 =x\Pi_y-y\Pi_x+\frac{G(r)}{2},
 }
\label{eq:LB}
\end{equation}
and the magnetic Runge--Lenz-type vector
\begin{equation}
 \boxed{
 \bm K_B
 =\left(L_z^{(B)}+\frac{G}{2}\right)\rr
 +(\rr\cdot\bPi)\crossz{\rr}
 -R^2\crossz{\bPi}.
 }
\label{eq:KB}
\end{equation}

\begin{theorem}[Magnetic symmetry with a split central term]
\label{thm:magnetic-algebra}
For \eqref{eq:BG}--\eqref{eq:KB},
\begin{align}
 \{L_z^{(B)},H_B\}&=0,
\label{eq:LBH}\\
 \{\bm K_B,H_B\}&=4H_B\crossz{\rr},
\label{eq:KBH}\\
 \{K_{B,x},L_z^{(B)}\}&=-K_{B,y},
\label{eq:KBxL}\\
 \{K_{B,y},L_z^{(B)}\}&=K_{B,x},
\label{eq:KByL}\\
 \{K_{B,x},K_{B,y}\}&=-4R^2L_z^{(B)}-Q.
\label{eq:KBKB}
\end{align}
Thus $\bm K_B$ is conserved on $H_B=0$.  In the unshifted generators the
magnetic realization contains a central term.
\end{theorem}

\begin{proof}
Equation \eqref{eq:Gprime} guarantees cancellation of the radial magnetic
terms in $\{L_z^{(B)},H_B\}$.  Substituting
\eqref{eq:BG}--\eqref{eq:KB} into the magnetic brackets
\eqref{eq:magnetic-PB} then gives
\eqref{eq:KBH}--\eqref{eq:KBKB} by direct calculation.
\end{proof}

The shifted generator
\begin{equation}
 \widetilde L_z
 =L_z^{(B)}+\frac{Q}{4R^2}
\label{eq:Lshift}
\end{equation}
restores the standard algebra.  Namely,
\begin{equation}
 \widetilde J_0=\widetilde L_z,
 \qquad
 \widetilde J_1=\frac{K_{B,x}}{2R},
 \qquad
 \widetilde J_2=\frac{K_{B,y}}{2R}
\end{equation}
satisfy the same $\mathfrak{so}(2,1)$ brackets as
\eqref{eq:so21}.  The central term is therefore split (algebraically
trivial): it is removed by an affine shift of the compact generator rather
than defining a new nonisomorphic Lie algebra.  This is compatible with the
central-extension language used for the spherical magnetic predecessor
\cite{BhandariCrescimanno}.  There, too, the quantum algebra is described as
unchanged up to a shift by a central element; here we display the affine shift
explicitly because the resulting standard $\mathfrak{so}(2,1)$ Casimir
controls the circle--horocycle--hypercycle transition.

\subsection{Charge-reversing inversion}

The nonmagnetic inversion preserved the moment map.  In the magnetic problem
orientation reversal forces the field parameter to change sign.  To state the
result precisely, display the dependence on the magnetic parameter explicitly.  Write
\begin{equation}
 B_Q(r)=-\frac{Q}{(r^2-R^2)^2},
 \qquad
 G_Q(r)=\frac{Q}{r^2-R^2},
\label{eq:BQGQ-explicit}
\end{equation}
and denote the corresponding quantities by
$L_{z,Q}^{(B)}$, $\bm K_{B,Q}$,
$\widetilde L_{z,Q}$, and $\mathcal C_Q$.

\begin{theorem}[Charge-reversing magnetic inversion duality]
\label{thm:magnetic-inversion}
Let $\Phi$ be the cotangent lift of circular inversion,
\begin{equation}
 \rr'=\frac{R^2}{r^2}\rr,
 \qquad
 \bPi'=\frac{r^2}{R^2}
 \left(
  \bPi-2\frac{\rr\cdot\bPi}{r^2}\rr
 \right).
\label{eq:magnetic-inversion-map}
\end{equation}
Then inversion exchanges the magnetic systems with parameters $Q$ and $-Q$.
More precisely,
\begin{equation}
 \boxed{
 \Phi^*\!\left(B_Q(r')\,dx'\wedge dy'\right)
 =B_{-Q}(r)\,dx\wedge dy,
 }
\label{eq:magnetic-twoform-inversion}
\end{equation}
so $\Phi$ is a symplectomorphism from the magnetic phase space with parameter
$-Q$ to that with parameter $Q$.  Moreover,
\begin{equation}
 \boxed{
 H_{B,Q}\circ\Phi
 =\frac{r^4}{R^4}H_{B,-Q},
 }
\label{eq:magnetic-H-inversion}
\end{equation}
and the shifted moment-map components obey
\begin{equation}
 \boxed{
 \widetilde L_{z,Q}\circ\Phi=\widetilde L_{z,-Q},
 \qquad
 \bm K_{B,Q}\circ\Phi=\bm K_{B,-Q}.
 }
\label{eq:magnetic-momentmap-inversion}
\end{equation}
Consequently,
\begin{equation}
 \boxed{
 \mathcal C_Q\circ\Phi=\mathcal C_{-Q}.
 }
\label{eq:magnetic-casimir-inversion}
\end{equation}
In particular, $\Phi$ exchanges exterior and punctured-interior zero-energy
magnetic trajectories while reversing the sign of $Q$.  If $z'=\Phi(z)$,
then on $H_{B,-Q}(z)=0$,
\begin{equation}
 \boxed{
 \Phi_*X_{H_{B,-Q}}(z)
 =\frac{r'^4}{R^4}X_{H_{B,Q}}(z')
 =\frac{R^4}{r^4}X_{H_{B,Q}}(z'),
 }
\label{eq:magnetic-flow-inversion}
\end{equation}
so the two flows are intertwined up to a positive reparametrization of time.
\end{theorem}

\begin{proof}
Circular inversion reverses orientation and has Jacobian determinant
\begin{equation}
 \det D\mathcal I=-\frac{R^4}{r^4}.
\end{equation}
Since $r'=R^2/r$,
\begin{equation}
 B_Q(r')
 =-\frac{Qr^4}{R^4(R^2-r^2)^2}.
\end{equation}
Multiplying by the transformed area form gives
\begin{equation}
 \Phi^*\!\left(B_Q(r')\,dx'\wedge dy'\right)
 =\frac{Q}{(R^2-r^2)^2}\,dx\wedge dy
 =B_{-Q}(r)\,dx\wedge dy,
\end{equation}
which proves \eqref{eq:magnetic-twoform-inversion}.  The kinetic term and
scalar potential transform exactly as in the nonmagnetic case,
\begin{equation}
 \bPi'^2=\frac{r^4}{R^4}\bPi^2,
 \qquad
 \frac{1}{(R^2-r'^2)^2}
 =\frac{r^4}{R^4}\frac{1}{(R^2-r^2)^2},
\end{equation}
and hence \eqref{eq:magnetic-H-inversion} follows.

Set
\begin{equation}
 \ell=x\Pi_y-y\Pi_x,
 \qquad
 s=\rr\cdot\bPi.
\end{equation}
Under \eqref{eq:magnetic-inversion-map},
\begin{equation}
 \ell'=\ell,
 \qquad
 s'=-s,
 \qquad
 G_Q(r')=\frac{r^2}{R^2}G_{-Q}(r).
\label{eq:magnetic-inversion-elementary}
\end{equation}
Substitution into
$\widetilde L_{z,Q}=\ell+G_Q/2+Q/(4R^2)$ gives
\begin{equation}
 \widetilde L_{z,Q}(\rr',\bPi')
 =\widetilde L_{z,-Q}(\rr,\bPi).
\end{equation}
For the vector generator, use the two-dimensional decomposition
\begin{equation}
 \bPi
 =\frac{s}{r^2}\rr
 +\frac{\ell}{r^2}\crossz{\rr}.
\label{eq:magnetic-momentum-decomposition}
\end{equation}
A direct substitution of
\cref{eq:magnetic-inversion-map,eq:magnetic-inversion-elementary} into
\eqref{eq:KB}, followed by
\eqref{eq:magnetic-momentum-decomposition}, yields
\begin{equation}
 \bm K_{B,Q}(\rr',\bPi')
 =\bm K_{B,-Q}(\rr,\bPi).
\end{equation}
This proves \eqref{eq:magnetic-momentmap-inversion};
\eqref{eq:magnetic-casimir-inversion} then follows immediately from the
quadratic definition of the shifted Casimir.

Finally, $\Phi$ is an involutive symplectomorphism between the $-Q$ and $Q$
magnetic phase spaces.  Therefore
\begin{equation}
 \Phi_*X_{H_{B,-Q}}
 =X_{H_{B,-Q}\circ\Phi}
\end{equation}
when the right-hand side is evaluated on the $Q$ phase space.  Applying
\eqref{eq:magnetic-H-inversion} with $Q$ replaced by $-Q$ gives a positive
factor $r'^4/R^4$; the term proportional to the Hamiltonian vanishes on the
zero-energy shell, proving \eqref{eq:magnetic-flow-inversion}.
\end{proof}

\begin{remark}
The sign reversal of $Q$ is forced geometrically by orientation reversal:
intrinsically the constant hyperbolic field is
$b_{\HH}=-Q/(4R^4)$, so circular inversion exchanges the two field
orientations.  Because \eqref{eq:magnetic-casimir-inversion} depends only on
the shifted moment map, the circle--horocycle--hypercycle regime is unchanged
under $Q\leftrightarrow -Q$.
\end{remark}

\subsection{Casimir and orbit equation}

The shifted algebra supplies the invariant that will classify the orbit
geometry.

\begin{proposition}[Magnetic Casimir]
\label{prop:magnetic-casimir}
The shifted quadratic Casimir is
\begin{equation}
 \mathcal C_Q
 =\frac{\bm K_B^2}{4R^2}
 -\left(L_z^{(B)}+\frac{Q}{4R^2}\right)^2.
\label{eq:CQ}
\end{equation}
It satisfies
\begin{equation}
 \boxed{
 \mathcal C_Q
 =\frac{m(R^2-r^2)^2}{2R^2}H_B
 +\frac{m\alpha}{2R^2}
 -\frac{Q^2}{16R^4}.
 }
\label{eq:CQidentity}
\end{equation}
On $H_B=0$,
\begin{equation}
 \boxed{
 \mathcal C_Q
 =\frac{m\alpha}{2R^2}
 -\frac{Q^2}{16R^4}.
 }
\label{eq:CQshell}
\end{equation}
\end{proposition}

\begin{proof}
Direct expansion of \eqref{eq:CQ} using \eqref{eq:LB} and
\eqref{eq:KB} yields \eqref{eq:CQidentity}; restriction to the zero-energy
shell gives \eqref{eq:CQshell}.
\end{proof}

\begin{proposition}[Magnetic orbit equation]
\label{prop:magnetic-orbit}
The magnetic generators obey
\begin{equation}
 \boxed{
 \bm K_B\cdot\rr
 =L_z^{(B)}(r^2+R^2)+\frac{Q}{2}.
 }
\label{eq:magnetic-Kdotr}
\end{equation}
For $L_z^{(B)}\neq0$, the zero-energy trajectories therefore lie on circles
\begin{equation}
 r^2
 -\frac{\bm K_B}{L_z^{(B)}}\cdot\rr
 +R^2+\frac{Q}{2L_z^{(B)}}=0.
\label{eq:magnetic-circle}
\end{equation}
Unlike the $Q=0$ geodesics, these circles are not generally orthogonal to
$r=R$.
\end{proposition}

\begin{proof}
Take the scalar product of \eqref{eq:KB} with $\rr$ and use
$(\crossz{\bPi})\cdot\rr=-(x\Pi_y-y\Pi_x)$ together with
\eqref{eq:LB}.
\end{proof}

\subsection{Intrinsic hyperbolic Landau interpretation}

The previous formulas are algebraic.  The next proposition explains their
intrinsic geometric meaning: the radial Euclidean field is constant when
measured against hyperbolic area.

\begin{proposition}[Exact hyperbolic Landau realization]
\label{prop:landau-realization}
The magnetic two-form is a constant multiple of the hyperbolic area form:
\begin{equation}
 \mathcal B=B(r)\,dx\wedge dy
 =b_{\HH}\,d\mu_{\HH},
 \qquad
 \boxed{b_{\HH}=-\frac{Q}{4R^4}}.
\label{eq:constant-hyperbolic-B}
\end{equation}
Moreover, define
\begin{equation}
 \mathcal H_{\HH,B}
 =\frac{(R^2-r^2)^2}{4R^4}H_B
  +\frac{\alpha}{4R^4}.
\label{eq:hyperbolic-landau-transform}
\end{equation}
Then
\begin{equation}
 \boxed{
 \mathcal H_{\HH,B}
 =\frac{(R^2-r^2)^2\bPi^2}{8mR^4},
 }
\label{eq:hyperbolic-landau-H}
\end{equation}
which is the kinetic Hamiltonian of the constant-field Landau problem on
$\mathbb H^2_R$.  On $H_B=0$ its energy is fixed to
\begin{equation}
 \mathcal H_{\HH,B}=\frac{\alpha}{4R^4},
\label{eq:landau-energy}
\end{equation}
and the two Hamiltonian vector fields are related by
\begin{equation}
 X_{\mathcal H_{\HH,B}}
 =\frac{(R^2-r^2)^2}{4R^4}X_{H_B}
 \qquad\text{on }H_B=0.
\label{eq:landau-flow-equivalence}
\end{equation}
Thus the Newtonian zero-energy magnetic trajectories and the fixed-energy
hyperbolic Landau trajectories coincide as unparametrized curves.
\end{proposition}

\begin{proof}
Equation \eqref{eq:constant-hyperbolic-B} follows from
$d\mu_{\HH}=\Omega^2dx\wedge dy$ and
$B/\Omega^2=-Q/(4R^4)$.  Expanding
\eqref{eq:hyperbolic-landau-transform} cancels the scalar potential and gives
\eqref{eq:hyperbolic-landau-H}.  Finally, if
$a(r)=(R^2-r^2)^2/(4R^4)$, then
$X_{aH_B}=aX_{H_B}+H_BX_a$ for the magnetic Poisson bracket.  Restriction to
$H_B=0$ proves \eqref{eq:landau-flow-equivalence}.
\end{proof}

The Landau formulation also gives a coordinate-free form of the threshold.
On the shell $H_B=0$, the hyperbolic speed is
\begin{equation}
 v_{\HH}
 =\sqrt{\frac{2\mathcal H_{\HH,B}}{m}}
 =\sqrt{\frac{\alpha}{2mR^4}}.
\label{eq:hyperbolic-magnetic-speed}
\end{equation}
Since the intrinsic magnetic field is
$b_{\HH}=-Q/(4R^4)$, the magnetic Lorentz equation gives the magnitude of the
geodesic curvature as
\begin{equation}
 \boxed{
 |\kappa_g|
 =\frac{|b_{\HH}|}{m v_{\HH}}
 =\frac{|Q|}{R^2\sqrt{8m\alpha}},
 \qquad
 R|\kappa_g|
 =\frac{|Q|}{\sqrt{8m\alpha R^2}}.
 }
\label{eq:magnetic-geodesic-curvature}
\end{equation}
Thus $R|\kappa_g|<1$, $=1$, and $>1$ give respectively a hypercycle, a
horocycle, and a closed magnetic circle; the special value $Q=0$ gives
$\kappa_g=0$ and hence a Poincar\'e geodesic.  This is the intrinsic version
of the Casimir classification below
\cite{ComtetHouston,Comtet,BarrosLandau,Plyushchay2026}.

\subsection{Circle, horocycle, and hypercycle regimes}

The intrinsic geodesic curvature already predicts a transition.  The shifted
Casimir now recovers the same transition directly from the Newtonian
conserved quantities and includes the zero-field geodesic limit.

\begin{theorem}[Magnetic orbit trichotomy and geodesic limit]
\label{thm:magnetic-classification}
Consider a zero-energy magnetic trajectory in $D_R$.
For $L_z^{(B)}\neq0$, set
\begin{equation}
 \bm c=\frac{\bm K_B}{2L_z^{(B)}},
 \qquad
 \rho^2
 =|\bm c|^2-R^2-\frac{Q}{2L_z^{(B)}}.
\label{eq:magnetic-center-radius}
\end{equation}
Its supporting curve is the Euclidean circle
$|\rr-\bm c|^2=\rho^2$.  Its relation to the ideal circle $r=R$ is determined
entirely by the shifted Casimir:
\begin{equation}
 \Delta_{\partial D}
 :=4R^2|\bm c|^2-
 \left(R^2+|\bm c|^2-\rho^2\right)^2
 =\frac{4R^4}{\left(L_z^{(B)}\right)^2}\,\mathcal C_Q.
\label{eq:boundary-discriminant}
\end{equation}
Consequently,
\begin{align}
 Q^2>8m\alpha R^2
 &\quad\Longleftrightarrow\quad \mathcal C_Q<0:
 &&\text{a closed magnetic circle contained in $D_R$},
\label{eq:circle-regime}\\
 Q^2=8m\alpha R^2
 &\quad\Longleftrightarrow\quad \mathcal C_Q=0:
 &&\text{a horocycle tangent to $r=R$},
\label{eq:horocycle-regime}\\
 0<Q^2<8m\alpha R^2
 &\quad\Longleftrightarrow\quad \mathcal C_Q>0:
 &&\text{an open hypercycle meeting $r=R$ twice}.
\label{eq:hypercycle-regime}
\end{align}
At $Q=0$ the magnetic field vanishes and the subcritical family reaches its
zero-geodesic-curvature limit: the supporting curve is a Poincar\'e geodesic,
namely an orthogonal circle when $L_z^{(B)}\neq0$.
If $L_z^{(B)}=0$, the orbit equation reduces to the line
\begin{equation}
 \bm K_B\cdot\rr=\frac Q2.
\label{eq:magnetic-line}
\end{equation}
An actual interior trajectory of this type can occur only when
$\mathcal C_Q>0$.  For $Q\neq0$ it is a line-model representative of a
hypercycle, while for $Q=0$ it reduces to a diameter geodesic.
\end{theorem}

\begin{proof}
Completing the square in \eqref{eq:magnetic-circle} gives
\eqref{eq:magnetic-center-radius}.  The standard intersection discriminant
for the circle of center $\bm c$ and radius $\rho$ with the boundary circle of
radius $R$ is the left-hand side of \eqref{eq:boundary-discriminant}.  Using
\begin{equation}
 R^2+|\bm c|^2-\rho^2
 =2R^2+\frac{Q}{2L_z^{(B)}}
\end{equation}
and the definition \eqref{eq:CQ}, one obtains
\begin{align}
 \left(L_z^{(B)}\right)^2\Delta_{\partial D}
 &=R^2\bm K_B^2
   -\left(2R^2L_z^{(B)}+\frac Q2\right)^2\\
 &=4R^4\mathcal C_Q.
\end{align}
Positive, zero, and negative discriminant correspond respectively to two
intersections, tangency, and no intersection.  In the last case the existence
of a point of the physical orbit in $D_R$ forces the entire supporting circle
to lie inside $D_R$.  Substitution of \eqref{eq:CQshell} gives
\eqref{eq:circle-regime}--\eqref{eq:hypercycle-regime}.

When $L_z^{(B)}=0$, \eqref{eq:magnetic-Kdotr} gives
\eqref{eq:magnetic-line}.  Its distance from the origin is
$|Q|/(2|\bm K_B|)$.  The identity
\begin{equation}
 \mathcal C_Q
 =\frac{4R^2\bm K_B^2-Q^2}{16R^4}
 \qquad\left(L_z^{(B)}=0\right)
\end{equation}
shows that the line intersects the open disk only when $\mathcal C_Q>0$.
\end{proof}

The classification is the standard geometry of a constant magnetic field on
$\mathbb H^2$, but \cref{thm:magnetic-classification} provides its explicit
realization and coupling threshold in the singular off-center Newtonian
system.  Equations \eqref{eq:magnetic-geodesic-curvature} and
\eqref{eq:CQshell} show the same transition intrinsically and algebraically:
it is encoded both in the geodesic curvature and in the shifted Casimir that
closes the Newtonian magnetic symmetry.

\Cref{fig:magnetic-trichotomy} displays the three regimes at a common scale.
The plot is generated directly from \eqref{eq:CQshell} and
\eqref{eq:magnetic-center-radius}, with $R=m=1$, $\alpha=0.05$, and
$L_z^{(B)}=-0.5$.  The dotted curves are the complete Euclidean supporting
circles, while the thick portions are the physical curves lying in the open
Poincar\'e disk.  The missing boundary points emphasize that a hypercycle and
a horocycle are open trajectories in the hyperbolic geometry, whereas the
strong-field magnetic circle is genuinely closed inside the disk.

\begin{figure}[!tbp]
\centering
\includegraphics[width=1\linewidth]{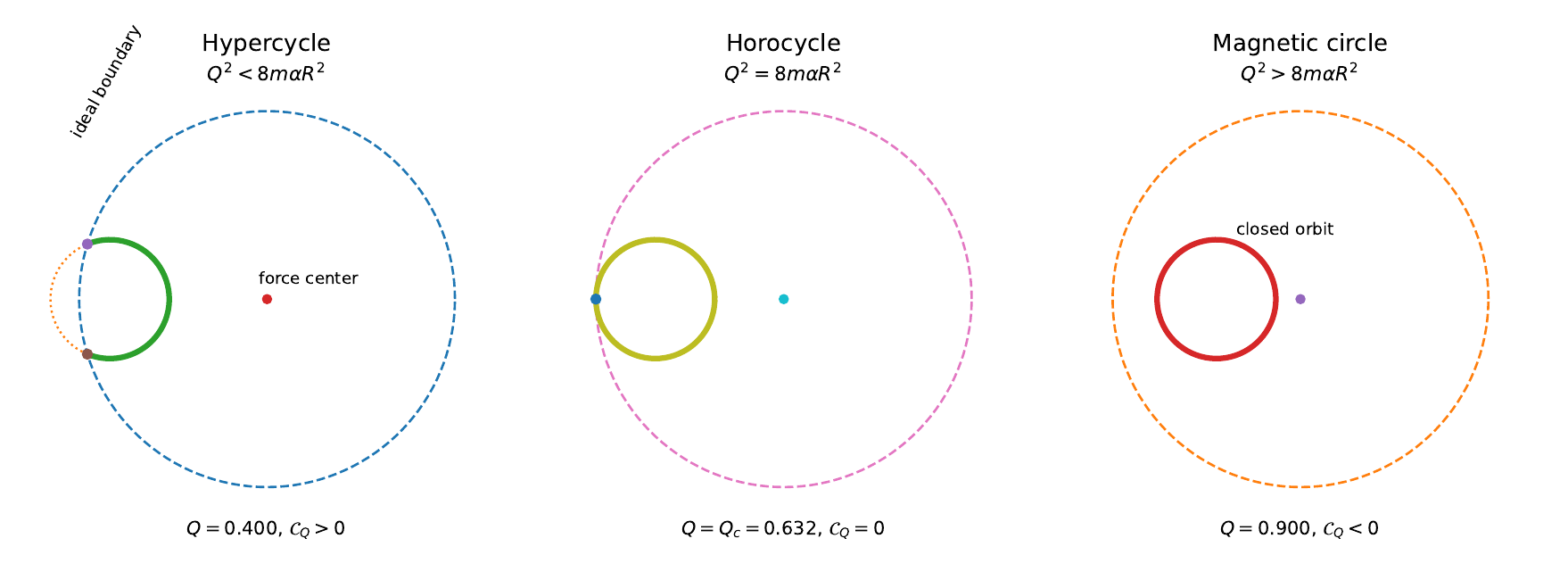}
\caption{Exact magnetic circle--horocycle--hypercycle trichotomy.  The dashed
large circles are the ideal boundaries $r=R$, dotted small circles are the
full Euclidean supporting circles, and thick curves are their physical parts
in $D_R$.  For $0<Q^2<8m\alpha R^2$ the orbit is a hypercycle with two
ideal endpoints; $Q=0$ is the geodesic limiting case; at
$Q^2=8m\alpha R^2$ it is a horocycle with one omitted tangency
point; and for $Q^2>8m\alpha R^2$ it is a closed magnetic circle wholly inside
the disk.  The plotted values are $Q=0.40$, $Q=\sqrt{0.4}$, and $Q=0.90$ for
$R=m=1$, $\alpha=0.05$, and $L_z^{(B)}=-0.5$.}
\label{fig:magnetic-trichotomy}
\end{figure}

Because the disk is contractible, there is no topological Dirac-monopole
charge on $D_R$ itself.  The term ``monopole'' is therefore best understood as
the analytic continuation of the spherical construction in
\cite{BhandariCrescimanno}, whereas the intrinsic hyperbolic description is a
constant Landau field.


\section{Discussion and conclusion}

The principal change in viewpoint is that the hyperbolic potential is no
longer treated as an isolated continuation guessed from the spherical case.
Under the fixed-energy assumptions of \cref{thm:intro-rigidity}, regular
radial Jacobi metrics of nonzero constant curvature force the two branches
\[
E-V(r)=\frac{\alpha}{(R^2\pm r^2)^2}.
\]
After shifting the energy origin to $E=0$, these are precisely
$V_{\pm}(r)=-\alpha/(R^2\pm r^2)^2$.
The radial Liouville equation behind this statement is classical, but its use
here isolates the structural reason that the spherical and hyperbolic
off-center models occur as a pair.

For the negative-curvature branch, the genuinely dynamical content begins
where the Poincare-disk observation ends.  The explicit quantities
$(L_z,\bm K)$ realize the $\mathfrak{so}(2,1)$ moment map in the original
Newtonian phase space.  Their identity
\[
\bm K\cdot\rr=L_z(r^2+R^2)
\]
gives the supporting circle algebraically, while the Casimir fixes its radius.
The cotangent lift of circular inversion preserves the moment map and relates
the exterior flow to the punctured interior flow after a positive time
change.  This is stronger than the geometric statement that inversion maps
orthogonal circles to themselves.  The boundary analysis also separates two
notions that the Poincare picture alone can obscure: $r=R$ is infinitely far
away in the Jacobi metric but is reached in finite Newtonian time.

The magnetic extension is similarly rigid.  If the full
orientation-preserving hyperbolic isometry group is to remain a symmetry of
the magnetic two-form, then the field must be a constant multiple of
hyperbolic area.  In disk coordinates this uniquely gives
\[
B(r)=-\frac{Q}{(R^2-r^2)^2}
\]
up to the single coupling $Q$.  The circle--horocycle--hypercycle trichotomy
is standard hyperbolic Landau geometry; the contribution here is the
symmetry-forced embedding of that geometry into the singular off-center
Hamiltonian, together with the explicit shifted moment map, coupling
threshold $Q^2=8m\alpha R^2$, and charge-reversing inversion map.

The scope is deliberately limited.  All hidden-symmetry statements are
on-shell at zero Newtonian energy, and the magnetic uniqueness theorem assumes
invariance under the full $SO^+(2,1)$ isometry group; weaker symmetry
requirements could admit other fields.  The short quantum discussion in
\cref{app:quantum} records a related coupling transform and boundary threshold
but is not used in the classical or magnetic results.

\section*{Generative AI disclosure}
OpenAI's ChatGPT was used as an auxiliary tool for exploratory algebra,
literature organization, computational checks, and manuscript editing.  All
claims and derivations in the final manuscript were independently checked by
the author, who takes responsibility for the content.

\FloatBarrier
\appendix
\section{Useful algebraic identities}\label{app:algebra}

For reference, the following identities hold identically on phase space:
\begin{align}
 \bm K^2
 &=4R^2L_z^2+(R^2-r^2)^2\pp^2,
\label{eq:appendix-K2}\\
 \bm K\cdot\rr
 &=L_z(r^2+R^2),
\label{eq:appendix-Kr}\\
 \{K_x,H\}&=-4yH,
 \qquad
 \{K_y,H\}=4xH.
\label{eq:appendix-KH}
\end{align}
The first two make the Casimir and orbit-circle derivations immediate.

For the magnetic system,
\begin{align}
 \bm K_B\cdot\rr
 &=L_z^{(B)}(r^2+R^2)+\frac{Q}{2},\\
 \frac{\bm K_B^2}{4R^2}
 -\left(L_z^{(B)}+\frac{Q}{4R^2}\right)^2
 &=\frac{m(R^2-r^2)^2}{2R^2}H_B
 +\frac{m\alpha}{2R^2}
 -\frac{Q^2}{16R^4}.
\end{align}
For $L_z^{(B)}\neq0$, the supporting-circle intersection discriminant is
\begin{equation}
 4R^2|\bm c|^2-
 \left(R^2+|\bm c|^2-\rho^2\right)^2
 =\frac{4R^4}{\left(L_z^{(B)}\right)^2}\mathcal C_Q,
\end{equation}
which gives the magnetic circle--horocycle--hypercycle classification,
including the zero-field geodesic limit, without solving the equations of
motion.

For the charge-reversing inversion theorem, if
$\ell=x\Pi_y-y\Pi_x$ and $s=\rr\cdot\bPi$, then
\begin{align}
 \ell'&=\ell,
 &s'&=-s,
 &G_Q(r')&=\frac{r^2}{R^2}G_{-Q}(r),\\
 \widetilde L_{z,Q}(\rr',\bPi')
 &=\widetilde L_{z,-Q}(\rr,\bPi),
 &\bm K_{B,Q}(\rr',\bPi')
 &=\bm K_{B,-Q}(\rr,\bPi).
\end{align}
These identities give
$\mathcal C_Q\circ\Phi=\mathcal C_{-Q}$ immediately and make the
$Q\leftrightarrow -Q$ duality directly checkable at the level of the moment
map.

\FloatBarrier

\section{A concise quantum side remark}
\label{app:quantum}

The classical results do not require quantization.  We nevertheless record one
short relation because it is easy to misstate.  With
\begin{equation}
\Omega(r)=\frac{2R^2}{R^2-r^2},
\qquad
g_{\mathbb H}=\Omega^2(dx^2+dy^2),
\end{equation}
the scalar Laplacian in two dimensions satisfies
\begin{equation}
\Delta_{\mathbb H}=\Omega^{-2}\Delta_{\mathbb R^2}
=\frac{(R^2-r^2)^2}{4R^4}\Delta_{\mathbb R^2}.
\label{eq:app-laplace}
\end{equation}
Hence
\begin{equation}
-\Delta_{\mathbb H}\psi=\lambda\psi
\end{equation}
is equivalent, as a differential equation, to
\begin{equation}
\left[-\Delta_{\mathbb R^2}
-\lambda\frac{4R^4}{(R^2-r^2)^2}\right]\psi=0.
\label{eq:app-stackel}
\end{equation}
This is a Stackel coupling transform rather than a unitary equivalence between
a fixed flat Hamiltonian and the hyperbolic Laplacian
\cite{KalninsMillerPost,MillerPostWinternitz}.

For the naive flat operator
\begin{equation}
\widehat H_{\alpha}^{\mathrm{flat}}
=-\frac{\hbar^2}{2m}\Delta_{\mathbb R^2}
-\frac{\alpha}{(R^2-r^2)^2},
\end{equation}
its zero-energy equation matches \eqref{eq:app-stackel} when
\begin{equation}
\lambda=\frac{m\alpha}{2\hbar^2R^4}.
\label{eq:app-lambda-alpha}
\end{equation}
Since the positive hyperbolic Laplacian has continuum threshold
$\lambda=1/(4R^2)$, this correspondence sends that edge to
\begin{equation}
\alpha_c=\frac{\hbar^2R^2}{2m}.
\label{eq:app-alpha-critical}
\end{equation}
The same value follows from the boundary indicial equation.  With
$\delta=R-r$,
\begin{equation}
-\frac{\hbar^2}{2m}\psi''
-\frac{\alpha}{4R^2\delta^2}\psi\simeq0,
\end{equation}
so $\psi\sim\delta^\beta$ gives
\begin{equation}
\beta_\pm=\frac12\left(1\pm
\sqrt{1-\frac{2m\alpha}{\hbar^2R^2}}\right).
\end{equation}
Thus \eqref{eq:app-alpha-critical} is the transition to logarithmically
oscillatory inverse-square behavior.  This is a statement about the coupling
transform and local boundary model, not a claim that the naive flat operator
is unitarily equivalent to $-\Delta_{\mathbb H}$.  The actual Euclidean-measure
representative obtained by the unitary map $f=\Omega\psi$ is
\begin{equation}
A_E=-\Omega^{-1}\Delta_{\mathbb R^2}\Omega^{-1},
\end{equation}
which contains variable-coefficient and first-order terms.  Questions of
self-adjoint extension for the naive inverse-square operator are separate
\cite{KalfEtAl,Case}.

\FloatBarrier

\end{document}